\documentclass{article}

\usepackage[T1]{fontenc}
\usepackage[utf8]{inputenc}
\usepackage[english]{babel}
\usepackage{authblk}

\usepackage{url}

\usepackage{breakurl}

\usepackage{amsmath,amssymb,amsthm,xspace}
\usepackage{tikz}
\usepackage{thm-restate}
\usepackage[final,pdfpagelabels=false,pdfpagelayout=SinglePage,bookmarksnumbered,hyperindex,colorlinks,linkcolor=black,citecolor=black,urlcolor=black,breaklinks]{hyperref}

\hypersetup{pdftitle={Effective guessing has unlikely consequences},pdfauthor={András Z. Salamon, Michael Wehar},pdfsubject={Effective guessing has unlikely consequences}}

\title{Effective Guessing Has Unlikely Consequences}

\author[1]{Andr\'as Z. Salamon}
\author[2]{Michael Wehar}
\affil[1]{
  School of Computer Science, University of St Andrews, UK
  \href{mailto:Andras.Salamon@st-andrews.ac.uk}{Andras.Salamon@st-andrews.ac.uk}
}
\affil[2]{
  Computer Science Department, Swarthmore College, USA
  \href{mailto:mwehar1@swarthmore.edu}{mwehar1@swarthmore.edu}
}
\date{}

\newcommand{\cc}[1]{\ensuremath{\mathsf{#1}}}
\newcommand{\dtime}{\ensuremath{\cc{DTIME}}\xspace}

\newcommand{\ntime}{\ensuremath{\cc{NTIME}}\xspace}

\newcommand{\ntigu}{\ensuremath{\cc{NTIGU}}\xspace}
\newcommand{\PTIME}{\ensuremath{\cc{P}}\xspace}
\newcommand{\NPTIME}{\ensuremath{\cc{NP}}\xspace}

\newcommand{\polylog}{\ensuremath{\mathrm{polylog}}}
\newcommand{\tuple}[1]{\ensuremath{\langle{#1}\rangle}}
\newcommand{\eps}{\varepsilon}
\newcommand{\zap}[1]{}

\newcommand{\cp}[1]{\textsc{#1}}
\newcommand{\sat}{\cp{SAT}\xspace}

\newcommand{\N}{\mathbb{N}}
\newcommand{\size}[1]{\lvert{#1}\rvert}
\newcommand{\tiwi}{\ensuremath{\cc{TIWI}}\xspace}

\newcommand{\shortdash}{\hbox{-}}
\newcommand{\pad}[2]{\text{pad}_{#1}\shortdash #2}

\newcommand{\xypair}{\tuple{x,y}}

\newcommand{\set}[2]{\{ \; #1 \; \vert \; #2 \; \}}

\newcommand{\proj}[2]{#1[#2]}

\newcommand{\starred}{^{*}}

\newtheorem{proposition}{Proposition}
\newtheorem{lemma}[proposition]{Lemma}
\newtheorem{theorem}[proposition]{Theorem}
\newtheorem{corollary}[proposition]{Corollary}
\newtheorem{conjecture}[proposition]{Conjecture}
\newtheorem{definition}[proposition]{Definition}

\newtheorem{note}{Note}
\newtheorem{example}{Example}

\begin{document}

\maketitle

\begin{abstract}
A classic result of Paul, Pippenger, Szemer\'edi and Trotter ~states that $\dtime(n) \subsetneq \ntime(n)$.
The natural question then arises: could the inclusion $\dtime(t(n)) \subseteq \ntime(n)$ hold for some superlinear time-constructible function $t(n)$?
If such a function $t(n)$ does exist, then there also exist effective nondeterministic guessing strategies to speed up deterministic computations.
In this work, we prove limitations on the effectiveness of nondeterministic guessing to speed up deterministic computations by showing that the existence of effective nondeterministic guessing strategies would have unlikely consequences.
In particular, we show that if a subpolynomial amount of nondeterministic guessing could be used to speed up deterministic computation by a polynomial factor, then $\PTIME \subsetneq \ntime(n)$.
Furthermore, even achieving a logarithmic speedup at the cost of making every step nondeterministic would show that $\sat \in \ntime(n)$ under appropriate encodings.
Of possibly independent interest, under such encodings we also show that $\sat$ can be decided in $O(n\log n)$ steps on a nondeterministic multitape Turing machine, improving on the well-known $O(n(\log n)^c)$ bound for some constant but undetermined exponent $c \ge 1$.%
\end{abstract}

\section{Dedication for Alan L.~Selman}
\label{sec:dedication}

This paper is dedicated to the memory of Alan L.~Selman.
The second author was a student in Professor Alan Selman's graduate course titled \textit{Introduction to the Theory of Computation} at University at Buffalo during Fall 2013.  Of the many fond memories of Professor Selman's instruction, his careful treatment of computational concepts and their technical details especially stood out.
During the course, Professor Selman's discussion of one particular topic is most clearly remembered and directly relates to this work, namely the topic of Linear Speedup Theorems for both deterministic and nondeterministic Turing machines (which can be found in Section 5.1 of \cite{Homer2011:computability}).  This discussion from Professor Selman is particularly relevant and motivating to this paper because the following investigates when additional nondeterminism could lead to improved efficiency, and the proof of the Linear Speedup Theorem for nondeterministic Turing machines (Theorem 5.3 of \cite{Homer2011:computability}) provides an example where the use of additional nondeterminism leads to a slightly faster simulation.

\section{Introduction}

How powerful is nondeterminism?
To make progress investigating this general philosophical question we have to consider a more focused technical question: \emph{when is it possible to replace some portion of a deterministic computation by nondeterministic guessing to reduce the total computation time?}

The Linear Speedup Theorem tells us that computations can be sped up by any constant factor by using larger tape alphabets~\cite[Theorem 2]{Hartmanis1965:computational}.
Conversely, the tight form of the Deterministic Time Hierarchy Theorem shows that it is not generally possible to achieve a speedup of more than a constant factor~\cite{Furer1982:tight}.
These classic results leave open the possibility that a computation could be further sped up by increasing some other resource such as nondeterminism.

Along with the example from Section~\ref{sec:dedication}, there are some cases where additional nondeterminism is known to speed up computation.
\sat is an example of a language for which, unless $\PTIME = \NPTIME$, we can speed up a deterministic decision procedure superpolynomially by instead guessing an assignment and verifying that the guess satisfies every clause.
Another example is deciding if an item occurs in a list.
With a random access model of computation, search can be sped up exponentially by guessing the position of the item in the input list and verifying this guess.
However, to the best of our knowledge, no general speed-up result has been proven for languages in \PTIME.

\section{Main Contributions}

We express our results in terms of complexity classes defined by joint bounds on time and nondeterminism.
Let $\ntigu(t(n),w(n))$ denote the class consisting of languages which can be decided by multitape nondeterministic Turing machines operating with time bound $O(t(n))$ and using at most $w(n)$ nondeterministic bits, for $n$-bit inputs.
(We follow prior usage with this definition~\cite{Fortnow:2016}.)

An \emph{effective guessing hypothesis} is a statement of the form: it is possible to speed up a computation by using more nondeterminism.
We present two main results conditional on two different effective guessing hypotheses, one somewhat stronger than the other.

Our first result is that if even a small polynomial speedup can be achieved by introducing a polylogarithmic amount of nondeterminism, then we could decide all of \PTIME using nondeterministic linear time.

\begin{theorem}
\label{theorem:p:nlin}
If there is some constant $c>1$ such that
 \[
  \dtime(n^c) \subseteq \ntigu(n,\polylog(n)),
\]
then $\PTIME \subsetneq \ntime(n)$.
\end{theorem}

Our second result greatly weakens the effective guessing hypothesis while still yielding a surprising conclusion.  In particular, the premise is weakened from polynomial to logarithmic speedup and from polylogarithmic to linear size witnesses.
Being able to speed up computation by a logarithmic factor, even at the cost of making essentially every step nondeterministic, would imply a breakthrough nondeterministic algorithm for \sat on multitape Turing machines.
This kind of effective guessing would allow us to overcome a barrier to progress that has stood for more than four decades.

\begin{restatable}{theorem}{weakegh}
\label{theorem:sat:nlin}
If $\dtime(n\log n) \subseteq \ntime(n)$, then $\sat \in \ntime(n)$.
\end{restatable}

Along the way to proving Theorem~\ref{theorem:sat:nlin} we also derive an improved upper bound for \sat.

\begin{restatable}{theorem}{satbound}
\label{theorem:sat}
$\sat \in \ntigu(n\log n, O(n/\log n))$.
\end{restatable}

\noindent
This improves a well-known upper bound
$
    \sat \in \ntime(n(\log n)^c)
$
for some constant $c\ge 1$ that is not explicitly stated in the literature, obtained either by a direct argument~\cite{Schnorr1978:satisfiability} or by using a random-access machine to perform the obvious guess-and-check algorithm in linear time, and using a standard simulation of RAMs by Turing machines~\cite{Santhanam2001:lower}.
Theorem~\ref{theorem:sat} allows us to take $c=1$.

\section{Relationship to Prior Work}

Some speedups are known to be impossible.
This is the case for nondeterministic computations, which cannot be sped up polynomially by using a moderate amount of advice.
In particular, Fortnow, Santhanam, and Trevisan showed that $\NPTIME \not\subseteq \ntime(n^c)/(\log n)^{1/2c}$ for all $c$~\cite{Fortnow2005:hierarchies}, and this result was extended by Fortnow and Santhanam to polynomial speedup and advice~\cite{Fortnow:2016}.

Our work fits into the long tradition of conditional separations and containments of complexity classes.
Previous work typically exploits classes that make nonuniform use of circuits.
This includes
Impagliazzo and Wigderson's result that if $\cc{E}$ requires circuits of exponential size for infinitely many input sizes, then $\cc{BPP} = \PTIME$~\cite{Impagliazzo1997:requires},
Fortnow and Santhanam's strengthening of the nondeterministic time hierarchy theorem in the presence of advice~\cite{Fortnow:2016},
and the unconditional separation of Williams of $\cc{ACC}$ circuits of polynomial size and $\cc{NEXP}$ which was achieved by proving two conditional containments and then combining them to yield a contradiction~\cite{Williams2014:nonuniform}.
In contrast, we focus here on subclasses of classical nondeterministic time classes which are defined by bounding the amount of nondeterminism.
Our Theorem~\ref{theorem:p:nlin} is also a significant extension of a result sketched by Bloch, Buss, and Goldsmith, weakening the effective guessing hypothesis used in their work from logarithmic to polylogarithmic nondeterminism~\cite{Bloch1994:how}.

Our work further relates to questions raised by the classical result that $\dtime(n) \subsetneq \ntime(n)$ of Paul et al.~\cite{Paul1983:determinism}.
This was obtained by assuming that $\dtime(n)$ and $\ntime(n)$ coincide, and then trading off an increase in alternations to obtain a speedup which contradicts a hierarchy theorem.
Beginning with this strict containment, it is then natural to consider whether $\dtime(t(n)) \subseteq \ntime(n)$ for any superlinear time-constructible function $t(n)$.
Our Theorem~\ref{theorem:sat:nlin} demonstrates that a positive answer to this question for even a mildly superlinear function such as $t(n) = n \log n$ would lead to a breakthrough for \sat.
Furthermore, from a form of the nondeterministic time hierarchy theorem~\cite[Corollary 2.3]{Zak:1983}, the strict complexity class containment $\dtime(n \log n) \subsetneq \ntime(n \log n)$ would straightforwardly follow.
Although Paul et al.~\cite{Paul1983:determinism} showed that the strict containment $\dtime(t(n)) \subsetneq \ntime(t(n))$ holds for $t(n) = n$, and Santhanam extended this to any $t(n) = o(n\lg\starred n)$~\cite[Theorem 2.5]{Santhanam:2000}, such a result is not known for functions $t(n)$ that grow at least as fast as $n\lg\starred n$.
(The iterated logarithm $\lg^{*} x$ is the minimal height of a tower of $2$s representing a number at least as large as $x$.)
Note that we do not use the alternation-trading technique from~\cite{Paul1983:determinism} in our work.

\section{Overview of Paper: Intuitions Behind Our\\ Arguments}

In Section~\ref{sec:prelim}, we first define time-witness classes as a technically convenient method of dealing with computations that limit the amount of nondeterminism.
These classes have similarities with advice classes, and essentially treat the guess as part of the input.
Our use of the existential projection allows for straightforward accounting of the nondeterministic bits when composing simulations.

In Section~\ref{sec:egh:strong}, we prove Theorem~\ref{theorem:p:nlin}.
This requires some machinery to precisely relate the speedup in each simulation step to the increase in witness size.
Our arguments in this section rely on subpolynomial functions being closed under composition, and a Strong Speedup Lemma (Lemma~\ref{lemma:speedup:strong}) for exact witness size bounds.

In Section~\ref{sec:egh:weak}, we prove Theorem~\ref{theorem:sat:nlin}.
The key is the Weak Speedup Lemma (Lemma~\ref{lemma:speedup:weak}) which allows witness size bounds up to arbitrary constant factors.
This lemma allows us to transfer an effective guessing hypothesis from deterministic to nondeterministic computations.
We provide an upper bound on the number of distinct variables that can be contained in an $n$-bit SAT instance when using a reasonable encoding.
This leads to a more precise upper bound for the time taken to decide \sat on a nondeterministic Turing machine (Theorem~\ref{theorem:sat}).
We will also need to make more precise the classical time upper bounds for sorting on a deterministic Turing machine.

Theorem~\ref{theorem:p:nlin} relies on the Strong Speedup Lemma (Lemma~\ref{lemma:speedup:strong}) and Theorem~\ref{theorem:sat:nlin} relies on the Weak Speedup Lemma (Lemma~\ref{lemma:speedup:weak}).
These lemmas are closely related but use incomparable hypotheses so require separate proofs.
Both speedup lemmas use the same intuition: if we assume some form of effective guessing hypothesis, then we can apply that hypothesis to speed up the deterministic verification step of a guess-and-check computation.
Our proofs make this intuition precise by using the time-witness class definitions to ensure that the increase in witness size is appropriately bounded.

In Section~\ref{sec:misc}, we discuss some final thoughts and outline directions for further work.

\section{Preliminaries}
\label{sec:prelim}

With $\N$ we mean the set $\{0,1,2,\dots\}$.
We assume a fixed alphabet $\Sigma=\{0,1\}$ throughout.
We also use the notation $\lg x = \log_2 x$ throughout.
For $x, y \in \Sigma\starred$, the expression $\tuple{x,y}$ simply denotes the bits of $x$ followed by those of $y$, also known as concatenation.
This guarantees associativity: $\tuple{x,\tuple{y,z}} = \tuple{\tuple{x,y},z}$.
For a word $x \in \Sigma\starred$, we denote by $\size{x}$ the number of symbols in $x$, which may be $0$ if $x$ is the empty word.
Hence $\size{\tuple{x,y}} = \size{x} + \size{y}$.
In a slight abuse of notation, when $f \colon \N \to \N$ is a function we will often write $f(n)$ to emphasize this fact, rather than just $f$, and when $c$ is a constant, we will sometimes use $c$ to denote the constant function $c(n) = c$.
With SAT we mean the Boolean Satisfiability problem for Boolean formulas in conjunctive normal form (CNF).

Our exposition is based on the concept of existential projection.
We will use existential projections to define classes of languages with combined time and witness size bounds.
This approach simplifies the bookkeeping required to track witness sizes in composed simulations, yet these classes are closely related to more familiar complexity classes.
We first define the notion and then discuss some consequences and an example.

\begin{definition}
\label{def:ep}
Given a language $L \in \Sigma^{*}$ and a function $w \colon \N \to \N$, the \emph{existential projection} of $L$ by $w$ is the language
\[
  L[w(n)] = \{ x \mid x \in \Sigma\starred, \, \exists y \in \Sigma^{w(\size{x})} \, \tuple{x,y} \in L \}.
\]
Further, for functions $f(n)$ and $g(n)$ let $L[f(n),g(n)] = (L[f(n)])[g(n)]$ as a convenient notation for composition of existential projections.
\end{definition}

The witness size function $w$ represents some portion of each word in the language which is set aside to record choices made by a nondeterministic computation.
The existential projection then removes this portion of each word.

We now make precise how Definition~\ref{def:ep} affects witness size changes in composition of existential projections.


\begin{lemma}
\label{lemma:project:composed}
For any language $L \in \Sigma\starred$, the identity
\[
  L[f(n),g(n)] = L[g(n) + f(n + g(n))]
\]
holds for all functions $f,g \colon \N \to \N$.
\end{lemma}


\begin{proof}
 %
Suppose $x \in L[f(n),g(n)] = (L[f(n)])[g(n)]$.
Then there is some $y \in \Sigma\starred$ such that $\tuple{x,y} \in L[f(n)]$ and $\size{y} = g(\size{x})$.
Further, there is some $z \in \Sigma\starred$ such that $\tuple{x,\tuple{y,z}} = \tuple{\tuple{x,y},z} \in L$ such that $\size{z} = f(\size{\tuple{x,y}}) = f(\size{x}+\size{y}) = f(\size{x} + g(\size{x}))$.
Hence $\size{\tuple{y,z}} = \size{y} + \size{z} = g(\size{x}) + f(\size{x} + g(\size{x}))$ and so $x \in L[g(n) + f(n+g(n))]$.

For the converse, suppose $x \in L[g(n) + f(n+g(n))]$.
Hence there is some $\alpha \in \Sigma^{*}$ such that $\tuple{x,\alpha} \in L$ and $\size{\alpha} = g(\size{x}) + f(\size{x} + g(\size{x}))$.
Let $y$ be the prefix of $\alpha$ consisting of the first $g(\size{x})$ symbols, and $z$ be the remaining $f(\size{x} + g(\size{x}))$ symbols.
Then $\tuple{x,y} \in L[f(n)]$ and hence $x \in (L[f(n)])[g(n)] = L[f(n),g(n)]$.
This completes our proof.
\end{proof}

\begin{example}
From Lemma~\ref{lemma:project:composed} it follows that $L[n,n] = L[3n]$.
This is most easily illustrated via a figure, showing how a word in $L$ changes as we project out the witness.

\begin{figure}[ht]
\centering
\begin{tikzpicture}[inner sep=1pt]
\draw (0,2) -- (0.98,2);
\node[below] at (0.5,2-0.02) {$n$};
\draw (1.02,2) -- (1.98,2);
\node[below] at (1.5,2-0.02) {$n$};
\draw (2.02,2) -- (3.98,2);
\node[below] at (3,2-0.02) {$2n$};
\node[left] at (0-0.5,2) {$L$};
\draw (0,1) -- (0.98,1);
\node[below] at (0.5,1-0.02) {$n$};
\draw (1.02,1) -- (1.98,1);
\node[below] at (1.5,1-0.02) {$n$};
\node[left] at (0-0.5,1) {$L[n]$};
\draw (0,0) -- (0.98,0);
\node[below] at (0.5,-0.02) {$n$};
\node[left] at (0-0.5,0) {$L[n,n]$};
\end{tikzpicture}
\qquad
\begin{tikzpicture}[inner sep=1pt]
\draw (0,2) -- (0.98,2);
\node[below] at (0.5,2-0.02) {$n$};
\draw (1.02,2) -- (1.98,2);
\node[below] at (1.5,2-0.02) {$n$};
\draw (2.02,2) -- (3.98,2);
\node[below] at (3,2-0.02) {$2n$};
\node[right] at (4+0.5,2) {$L$};
\draw (0,0) -- (0.98,0);
\node[below] at (0.5,-0.02) {$n$};
\node[right] at (4+0.5,0) {$L[3n]$};
\end{tikzpicture}
\end{figure}
\end{example}


We now define time-witness classes in terms of existential projections.

\begin{definition}
\label{def:tiwi}
 %
\(
 \tiwi(t(n),w(n)) = \{ L[w(n)] \mid L \in \dtime(t(n)) \}.
\)
 %
\end{definition}

\begin{note}
Our definition of $\tiwi$ has similarities with the classical definition of advice classes~\cite{Karp1982:Turing}.
In the advice setting, the witness values are uniquely determined by a possibly uncomputable oracle function.
In contrast the witness values are not determined with our notion of existential projection, and we therefore avoid undecidable languages in our $\tiwi$ classes.
\end{note}

We need to be careful to account for the total nondeterminism when composing two or more nondeterministic simulations; such compositions are crucial for the proofs of our results.
We have chosen to define time-witness classes via the existential projection because this notation assists in explicitly keeping track of witness size scaling in compositions.
Existential projections also allow us to control the overhead of encoding.
Such detailed bookkeeping becomes necessary when the number of compositions in an argument is allowed to grow with the instance size.

For a function $w \colon \N \to \N$, let
 \[
  \tiwi(t(n),O(w(n))) = \bigcup_{c>0} \tiwi(t(n),cw(n)).
 \]
As usual, $\polylog(n)$ denotes the class of functions
\[
  \bigcup_{c>0} \{ f(n) \colon \N \to \N \mid f(n) = O((\lg n)^c) \}.
\]
As a notational convenience, if $\phi(n)$ is a logical expression in which the variable $n$ occurs free and there are no other free variables, then we say that $\phi(n)$
holds \emph{eventually} if there exists some $n_0 \in \N$ such that $\phi(n)$ is a true sentence for all $n \in \N$ such that $n \ge n_0$.

The following lemma will simplify several proofs.
This result allows us to ignore minor differences in the witness size functions when comparing two time-witness classes, and instead to focus on how they behave asymptotically.

\begin{lemma}
\label{lemma:tiwi:increase}
Suppose $t \colon \N \to \N$ is a function such that $t(n) = \Omega(n)$.
If $v(n) \le w(n)$ eventually, then $\tiwi(t(n),v(n)) \subseteq \tiwi(t(n),w(n))$.
 %
\end{lemma}

\begin{proof}
Suppose $K \in \tiwi(t(n),v(n))$.
Hence there is some $L \in \dtime(t(n))$ such that $K = L[v(n)]$.
Let $L' = \{ \tuple{\tuple{x,y},z} \mid \tuple{x,y} \in L, \size{z} + \size{y} = w(\size{x}) \}$.
Note that $L'$ does not include any word $\tuple{\tuple{x,y},z}$ for which $v(\size{x}) > w(\size{x})$.
Since $v(n) \le w(n)$ eventually, the language $L'$ does include such words for all sufficiently large $\size{x}$, and therefore $K$ and $L'[w(n)]$ only differ in at most finitely many words up to some threshold size $n_0$, and will be the same for all inputs of size $n_0$ and greater.
Now we can decide $L'$ by using a decider for $L$ while ignoring any additional part of the input.
The overhead of setting up the decider for $L$ is linear, and because $t(n) = \Omega(n)$ we have that $L' \in \dtime(t(n))$.
Hence $L'[w(n)] \in \tiwi(t(n),w(n))$.

Now we can further augment the decider for $L'[w(n)]$ with a brute force simulation which deterministically checks all possible witnesses for an input $x \in K$ if $\size{x}$ is below the threshold size $n_0$ where $v(n) \le w(n)$ for all $n \ge n_0$.
This introduces a large constant factor into the simulation, but this is taken care of by the time bound being $O(t(n))$.
Therefore $K \in \tiwi(t(n),w(n))$.
 %
\end{proof}

For functions $w,t \colon \N \to \N$, we say that $w(n)$ is computable in $t(n)$ time if there is a deterministic Turing machine that when given $x \in \Sigma\starred$, in at most $t(\size{x})$ steps writes a word on its output tape with precisely $w(\size{x})$ symbols.
For simplicity, we assume that all witness size functions in the following are computable within the provided time bounds.
By first computing the witness size function for the given input, a Turing machine can determine where a word finishes and the witness bits begin, thereby avoiding overhead for self-terminating encodings or separator characters in the alphabet.

$\ntigu(t(n),w(n))$ is the class of languages decidable by a
multitape Turing machine that takes at most $O(t(n))$ steps and uses at
most $w(n)$ nondeterministic bits on any input of length $n$ (for
instance, see~\cite{Fortnow:2016}).
The following results explain our choice of Definition~\ref{def:tiwi}, by relating the time-witness $\tiwi$ classes to the more familiar $\ntime$ (Lemma~\ref{lemma:tiwi:ntime}) and time-guess $\ntigu$ (Lemma~\ref{lemma:tiwi:ntigu}) classes.

First we consider the case where witness size is quite large, even possibly dominating the input size.


\begin{lemma}
\label{lemma:tiwi:ntime}
If $w(n) = \Omega(n)$, then $\tiwi(n,O(w(n))) = \ntime(w(n))$.
\end{lemma}

\begin{proof}
First suppose $c > 0$ and let $K \in \tiwi(n,c\cdot w(n))$.
Then there is some $L \in \dtime(n)$ such that $K = L[c\cdot w(n)]$.
Let $M$ be a deterministic Turing machine which decides $L$ in $O(n)$ steps.
We define a nondeterministic Turing machine $M'$ that given an $n$-bit input $x$, writes a copy of $x$ followed by a string $y$ consisting of $c\cdot w(n)$ bits chosen nondeterministically, and simulates $M$ with input $\tuple{x,y}$.
Then $M'$ takes $O(n + c\cdot w(n)) + O(c\cdot w(n)) = O(w(n))$ steps to decide whether $\tuple{x,y} \in L$, and therefore $K \in \ntime(w(n))$.

Now suppose $K \in \ntime(w(n))$.
Then there is a nondeterministic Turing machine $M'$ which decides $K$ in $O(w(n))$ steps.
Given an $n$-bit input $x$, we then can construct a Turing machine $M$ which records $O(w(n))$ nondeterministic bits for the steps taken by $M'$, and verifies in linear time in the size of the guessed sequence of actions whether $M'$ accepts.
Consider the language $L$ consisting of strings $\tuple{x,y}$ such that $x\in K$ and $\size{y} = O(w(n))$, where $y$ records the moves made by an accepting computation of $M'$ on input $x$.
Then $L \in \dtime(n)$ and $K = L[O(w(n))]$.
Hence $K \in \tiwi(n,O(w(n)))$.
 %
\end{proof}

If $k \ge 3$ then any $k$-tape Turing machine can be simulated by a two-tape machine with a logarithmic increase in time~\cite{Hennie1966:two}.
This slowdown does not affect Lemma~\ref{lemma:tiwi:ntime} as we are not trying to reduce the number of tapes: in the proof each inclusion increases the number of tapes.
This increase does not matter for the purpose of establishing the result, or for the applications where we use it, although it might be important for other applications of the technique where the number of tapes has to be more carefully controlled.

Lemma~\ref{lemma:tiwi:ntime} showed that $\tiwi$ classes for superlinear witness bounds lose the discrimination power of the $\ntigu$ classes.
However, our next result shows that $\tiwi$ and $\ntigu$ classes coincide for at most linear witness size and many common time bounds.

\begin{lemma}
\label{lemma:tiwi:ntigu}
Suppose that $w(n) \le n$ for all $n$, and that there exist constants $c\ge 1$ and $d\ge 0$ such that $t(n) = \Theta(n^c(\lg n)^d)$.
Then
 \(
  \tiwi(t(n),w(n)) = \ntigu(t(n),w(n)).
 \)
\end{lemma}

\begin{proof}
First let $K \in \tiwi(t(n),w(n))$.
Then there exists some language $L \in \dtime(t(n))$ such that $K = L[w(n)]$.
Suppose that $M$ is a deterministic Turing machine which decides $L$ in $O(t(n))$ steps.
We need to decide whether $x \in K$ using a nondeterministic machine $M'$.
$M'$ first copies $x$ to a tape (which will be the input tape of $M$), appends $w(\size{x})$ bits to this tape the values of which are determined nondeterministically, and then simulates $M$.
The simulation can be performed using a constant number of deterministic steps per step of $M$, so the total number of steps to decide whether $x \in K$ is at most $a\cdot t(\size{x}+w(\size{x}))+b\cdot w\size{x}$ for some constants $a,b$.
Moreover, $M'$ uses $w(\size{x})$ nondeterministic bits and accepts $x$ if, and only if, there is some $y \in \Sigma^{w(n)}$ such that $\tuple{x,y} \in L$.
This is equivalent to saying that $M'$ accepts $x$ iff $x \in K$, so $M'$ correctly decides $K$.

We now claim that if $T(n) \le a\cdot t(n+w(n))+b\cdot w(n)$ eventually for constants $a,b$, then the conditions guarantee that $T(n) = O(t(n))$.
Therefore $K \in \ntigu(t(n),w(n))$ by putting $T(n)$ as the largest number of steps taken by $M'$ to decide an input of $n$ bits.
To prove the claim, suppose there are constants $a,b$ such that $T(n) \le a\cdot t(n+w(n))+b\cdot w(n)$ eventually.
Since $t(n) = O(n^c(\lg n)^d)$, we have some constant $e$ so that $t(n) \le e\cdot n^c(\lg n)^d$ eventually.
Then $T(n) \le a\cdot t(n+w(n))+b\cdot w(n) \le a\cdot e(n+w(n))^c(\lg(n+w(n)))^d + b\cdot w(n) \le a\cdot e(2n)^c(1+\lg n)^d + b\cdot n \le a\cdot e\cdot (2n)^c(2\lg n)^d + b\cdot n = (a 2^c 2^d e)n^c(\lg n)^d + b\cdot n$ eventually.
As $c \ge 1$ and $d \ge 0$ it follows that $T(n) = O(n^c(\lg n)^d)$.
Now since $t(n) = \Omega(n^c(\lg n)^d)$, we have that $T(n) = O(t(n))$, and we are done.

For the other direction, suppose $K \in \ntigu(t(n),w(n))$.
There is then some nondeterministic Turing machine $M'$ which decides $K$ using $O(t(n))$ steps and with $w(n)$ bits of nondeterminism.
$M'$ accepts $x$ precisely when there is some $y \in \Sigma^{w(n)}$ representing the nondeterministic moves made by $M'$ in reaching an accepting state (within $O(t(n))$ steps).
We build a deterministic Turing machine $M$, which when given $x$ and the nondeterministic moves $y$ as input $\tuple{x,y}$, simulates $M'$ using a number of deterministic steps per step of $M'$ that is bounded by some constant $a$.
To achieve this we first compute $w(\size{x})$ in at most $O(t(\size{x})$ steps and store this on a unary tape, which we can then use to determine where $x$ ends and $y$ begins.
Hence $M$ on input $\tuple{x,y}$ takes at most $a\cdot t(\size{x}) + b\size{y}$ steps, so $O(t(\size{x}+\size{z}))$ steps.
We can now let language $L$ be the language of words $\tuple{x,y}$ accepted by $M$.
There is some $y\in \Sigma^{w(n)}$ such that machine $M$ accepts $\tuple{x,y}$ iff $M'$ accepts $x$.
Then $K = L[w(n)]$ and $L \in \dtime(t(n))$, so $K \in \tiwi(t(n),w(n))$.
 %
\end{proof}

In the proof of Lemma~\ref{lemma:tiwi:ntigu} each inclusion increases the number of tapes by one.
This does not affect our results but might restrict some applications.

While classes such as $\cc{NTIGU}(t(n),w(n))$ measure the time bound in terms of the input, $\tiwi(t(n),w(n))$ measures the time bound as a function of both the input and the witness.
For witnesses growing strictly faster than the size of the input, the two definitions can diverge where $\tiwi(t(n),O(w(n)))$ is not equal to $\ntigu(t(n),O(w(n)))$.
To see this, take $t(n) = n$ and $w(n) = n^2$.
We have that $\ntigu(n,O(n^2)) = \ntigu(n,O(n)) = \ntime(n)$ because additional guess bits beyond $O(t(n))$ do not help us.
However, $\ntime(n^2) = \tiwi(n,O(n^2))$ by Lemma~\ref{lemma:tiwi:ntime}.
Therefore, $\ntigu(n,O(n^2)) = \ntime(n) \subsetneq \ntime(n^2) = \tiwi(n,O(n^2))$ by the nondeterministic time hierarchy theorem.

Notice from the preceding discussion that as the witness size grows beyond the input size, the $\ntigu$ classes no longer capure new languages while the $\tiwi$ classes become equivalent to the coarser classical nondeterministic time classes.
Even though these two definitions can diverge when the witness size is larger than the input size, Lemma~\ref{lemma:tiwi:ntigu} allows us to use $\tiwi$ rather than $\ntigu$ in the subsequent discussion, because in this work we are interested in witnesses of moderate size, at most as large as the input size and computable within the given time bounds.
The choice of $\tiwi$ avoids technical difficulties arising from instance size blowup when composing simulations, and essentially amounts to preallocating the nondeterministic bits which are used in a computation and including them in the instance size.

\section{Strong Effective Guessing Would Imply Linear-Time Simulation of \texorpdfstring{$\PTIME$}{P}}
\label{sec:egh:strong}

In this section we show that a strong form of guessing would imply that all of $\PTIME$ is contained in $\ntime(n)$.
It appears to us that this is unlikely to be true.
Even though such an inclusion in turn would imply that $\PTIME \ne \NPTIME$, many other rather less likely consequences would also follow.
These include improving the current best $O(n^{2.37286})$ algorithm for multiplication of $n$ by $n$ matrices (see \cite{Alman2021:refined}) to $\tilde{O}(n^{2})$ time, reducing the time for general graph maximum matching from $\tilde{O}(n^{2.5})$ (see \cite{Vazirani2012:simplification}) to $\tilde{O}(n^{2})$, and reducing the $n^{k/c}$ time (for some $c$ such that $1 \le c < k$) to decide if an input graph contains a $k$-clique to $O(n\lg n)$ time, all achieved through the use of nondeterminism.
Yet it is not at all clear that allowing guessing could significantly speed up so many well-studied and disparate algorithms.
(Here we use the common convention that $\tilde{O}(t(n))$ denotes the class of functions $\bigcup_{c>0} O(t(n)(\lg t(n))^c)$.)

Informally, our argument for Theorem~\ref{theorem:p:nlin} works as follows.
We have defined the class $\tiwi(n,\lg n)$, which by Lemma~\ref{lemma:project:composed} can be regarded as the class of languages decided by nondeterministic machines that use linear time and $\lg n$ bits of nondeterminism.
We suppose that $\dtime(n^2) \subseteq \tiwi(n,\lg n)$.
By a padding argument it follows that $\dtime(n^4) \subseteq \tiwi(n^2,\lg n)$.
Now suppose $L \in \tiwi(n^2,\lg n)$; this means that there is a language $L' \in \dtime(n^2)$ such that $x \in L$ if there is a $y$ of length $\lg n$ and $xy \in L'$.
Again applying our hypothesis, this time to $L'$, we conclude via Lemmas~\ref{lemma:tiwi:increase} and \ref{lemma:tiwi:ntime} that $L \in \tiwi(n,\lg n)$.
We therefore conclude that $\dtime(n^4) \subseteq \tiwi(n,\lg n)$.
We can then use this step in an induction argument.
We now proceed with a formal version of this argument.

In preparation for our result, we need to demonstrate that time-witness classes are structurally well-behaved.
We first establish conditions which ensure that increases in witness size are kept reasonable when applying effective guessing.

\begin{lemma}
\label{lemma:speedup:strong:sub}
If $\dtime(t(n)) \subseteq \tiwi(t'(n),w(n))$ then
 \[
  \tiwi(t(n),w'(n)) \subseteq \tiwi(t'(n), w(n+w'(n)) + w'(n)).
 \]
\end{lemma}

\begin{proof}
Suppose that $\dtime(t(n)) \subseteq \tiwi(t'(n),w(n))$, and let $J$ be an arbitrary language in $\tiwi(t(n),w'(n))$.
By definition then there exists some language $K$ in $\dtime(t(n))$ such that $J = \proj{K}{w'(n)}$.
By our assumption there must exist some $L \in \dtime(t'(n))$ such that $K = \proj{L}{w(n)}$.
By Lemma~\ref{lemma:project:composed}, we have that $J = \proj{K}{w'(n)} = \proj{L}{w(n),w'(n)} = \proj{L}{w(n + w'(n)) + w'(n)}$.
Finally, we can conclude that $J \in \tiwi(t'(n), w(n+w'(n)) + w'(n))$.
\end{proof}

\begin{lemma}[Strong Speedup]
\label{lemma:speedup:strong}
Suppose $\dtime(t(n)) \subseteq \tiwi(t'(n),w(n))$.
For all functions $w' \colon \N \to \N$ for which there exists a constant $C$ such that ${w(n+w'(n))}$ $\le C\cdot w(n)$ eventually, we have that
 \[
  \tiwi(t(n),w'(n)) \subseteq \tiwi(t'(n), C \cdot w(n) + w'(n)).
 \]
\end{lemma}

\begin{proof}
Suppose first that the inclusion $\dtime(t(n)) \subseteq \tiwi(t'(n),w(n))$ holds, and let $K$ be a language in $\tiwi(t(n),w'(n))$.
Via Lemma~\ref{lemma:speedup:strong:sub} we can now conclude that
 \(
  K \in \tiwi(t'(n), w(n+w'(n)) + w'(n)).
 \)
From the properties of $w'$ and Lemma~\ref{lemma:tiwi:increase} it then follows that $K \in \tiwi(t'(n), C \cdot w(n) + w'(n))$.
\end{proof}

We continue with a useful amplification property of time-witness classes in the presence of effective guessing.
By analogy with superadditive functions (see~\cite{Book1970:time}), we say that a function $f \colon \N \to \N$ is \emph{weakly superadditive} if
$f(n+d) \ge f(n) + d$ for all $d,n \in \N$.
Note that any function $f(n) = n^c$, where $c\ge 1$, is weakly superadditive.

\begin{lemma}
\label{lemma:reverse:padding}
Let $f$ be a weakly superadditive function.
If
 \[
  \dtime(t(n)) \subseteq \tiwi(t'(n),w(n))
 \]
then
 \[
  \dtime(t(f(n))) \subseteq \tiwi(t'(f(n)),w(f(n))).
 \]
\end{lemma}

\begin{proof}
Suppose that $\dtime(t(n)) \subseteq \tiwi(t'(n),w(n))$.
Further, let
$K$ be an arbitrary language in $\dtime(t(f(n)))$.
Consider the function $B(n) = f(n) - n$.
Let
 \[
  \pad{B}{K} = \set{\tuple{1^k,x}}{k = B(\size{x}) \, \wedge \, x \in K}.
 \]
As $K \in \dtime(t(f(n)))$, we have $\pad{B}{K} \in \dtime(t(n))$.
By the assumption, $\pad{B}{K} \in \tiwi(t'(n),w(n))$.
Hence there is some $L \in \dtime(t'(n))$ such that $\pad{B}{K} = \proj{L}{w(n)}$.
Let
 \[
  L' = \set{\xypair}{y \in \Sigma^{w(f(\size{x}))} \, \wedge \, \tuple{1^{B(\size{x})},\xypair} \in L}.
 \]
Since $\tuple{1^{B(\size{x})},\tuple{x,y}} = \tuple{\tuple{1^{B(\size{x})},x},y}$, it follows that $K = L'[w(f(n))]$.
To show that $K \in \tiwi(t'(f(n)),w(f(n)))$, it then suffices to show that $L' \in \dtime(t'(f(n)))$.

Because of our choice of the function $B$ and since $L \in \dtime(t'(n))$, we can determine if $\xypair \in L'$, where $\size{y} = w(f(\size{x}))$, in time
\begin{align*}
  O(t'(\size{\tuple{1^{B(\size{x})},\xypair}}))
   &= O(t'(B(\size{x}) + \size{x} + \size{y})) \\
   &= O(t'(f(\size{x}) + \size{y})) \\
   &\leq O(t'(f(\size{x} + \size{y}))) \\
   &= O(t'(f(\size{\xypair}))).
\end{align*}
Therefore, $L' \in \dtime(t'(f(n)))$.
\end{proof}

We say that a function $f(n)$ is \emph{subpolynomial} if for every $c > 0$ we have that $f(n) = o(n^c)$,
and \emph{semihomogeneous} (see~\cite{Book1970:time}) if for any $d > 1$ there is a constant $C = C(d)$ such that eventually $f(dn) \le C \cdot f(n)$.
Note that any polylogarithmic function (such as $f(n) = (\lg n)^3$) is subpolynomial and also semihomogeneous.

This leads up to our first amplification argument, showing that a form of effective guessing with small witnesses can be amplified to yield a larger speedup at the cost of only a moderate amount of additional guessing.

\begin{lemma}
\label{lemma:reverse:iterative1}
Let $c \ge 1$ be a constant, and let $v(n)$ be a non-decreasing function that is subpolynomial, semihomogeneous, and increases infinitely often.
If $\dtime(n^c) \subseteq \tiwi(n,v(n))$ then for all $k \in \N$, $\dtime(n^{c^{k+1}}) \subseteq$ ${\tiwi(n,C^k \cdot v(n^{c^k}))}$ for some constant $C \ge 1$.
\end{lemma}

\begin{proof}
Suppose that
 \(
  \dtime(n^c) \subseteq \tiwi(n,v(n)).
 \)
We will show by induction that
 \[
  \dtime(n^{c^{k+1}}) \subseteq \tiwi(n,C^k \cdot v(n^{c^k}))
 \]
for all $k \in \N$.
The base case holds for $k=0$ since $\dtime(n^c) \subseteq \tiwi(n,v(n))$ is true by assumption.
For the inductive step, suppose that for some $k \ge 1$ we have
 \[
  \dtime(n^{c^k}) \subseteq \tiwi(n,C^{k-1} \cdot v(n^{c^{k-1}})).
 \]
Now apply Lemma~\ref{lemma:speedup:strong} to this inclusion with $w(n) = C^{k-1} \cdot v(n^{c^{k-1}})$, $w'(n) = C^{k-1} \cdot v(n^{c^k})$, $t(n) = n^{c^k}$, and $t'(n) = n$, to obtain
 \[
  \tiwi(n^{c^k}, C^{k-1} \cdot v(n^{c^k})) \subseteq \tiwi(n, C^k \cdot v(n^{c^{k-1}}) + C^{k-1} \cdot v(n^{c^k})).
 \]
We can do this because the properties of $v$ ensure that for any $\eps > 0$, eventually
\begin{align*}
 w(n + w'(n))
  &= C^{k-1} \cdot v((n + w'(n))^{c^{k-1}}) \\
  &\le C^{k-1} \cdot v(((1 + \eps) n)^{c^{k-1}}) \\
  &\le C^{k-1} \cdot C \cdot v(n^{c^{k-1}}) \\
  &= C\cdot w(n).
\end{align*}
As a second step, now apply Lemma~\ref{lemma:reverse:padding} to the same assumption, with $f(n) = n^{c^k}$, $w(n) = C^{k-1} \cdot v(n)$, $t(n) = n^c$, and $t'(n) = n$, to obtain
 \[
  \dtime(n^{c^{k + 1}}) \subseteq \tiwi(n^{c^k}, C^{k-1} \cdot v(n^{c^k})).
 \]
This allows us to conclude that
 \[
  \dtime(n^{c^{k+1}}) \subseteq \tiwi(n,C^k \cdot v(n^{c^{k-1}}) + C^{k-1} \cdot v(n^{c^k})),
 \]
and hence
 \[
  \dtime(n^{c^{k+1}}) \subseteq \tiwi(n,C^k \cdot v(n^{c^k}) (v(n^{c^{k-1}})/v(n^{c^k}) + 1/C)).
 \]
As $v$ increases infinitely often, by Lemma~\ref{lemma:tiwi:increase} we then have that
 \[
  \dtime(n^{c^{k+1}}) \subseteq \tiwi(n,C^k \cdot v(n^{c^k})),
 \]
which completes our proof.
\end{proof}

We now wrap up our second amplification argument into a theorem.

\begin{theorem}
\label{theorem:nlin}
If there exists a constant $c > 1$ and a subpolynomial function $v(n)$ such that $\dtime(n^c) \subseteq \tiwi(n,v(n))$, then $\PTIME \subsetneq \ntime(n) \subsetneq \NPTIME$.
\end{theorem}

\begin{proof}
Suppose that there exists $c > 1$ and a subpolynomial function $v(n)$ such that $\dtime(n^c) \subseteq \tiwi(n,v(n))$.
Since $v(n)$ is subpolynomial, so is $v(n^c)$ for any $c>0$.
By applying Lemma~\ref{lemma:reverse:iterative1}, we then have that for all $k \in \N$,
 \[
  \dtime(n^{c^k}) \subseteq \tiwi(n,w(n))
 \]
for some subpolynomial function $w(n)$.
Since $c > 1$,
 \(
  \lim_{k\to\infty}c^k = \infty.
 \)
By Lemmas~\ref{lemma:tiwi:increase} and~\ref{lemma:tiwi:ntime} we then have that for all $k \in \N$,
 \(
  \tiwi(n,w(n)) \subseteq \tiwi(n,O(n)) = \ntime(n).
 \)
It follows that $\PTIME \subseteq \ntime(n)$.
Further, $\PTIME \neq \ntime(n)$ can be shown by a standard padding argument applied to the nondeterministic time hierarchy theorem \cite{Zak:1983}.
\end{proof}

Via Lemma~\ref{lemma:tiwi:ntigu}, Theorem~\ref{theorem:p:nlin} is a corollary of Theorem~\ref{theorem:nlin} for the special case that $v(n)$ is a subpolynomial function that grows faster than any polylogarithmic function; $v(n) = (\lg n)^{\lg\lg n}$ is an example of such a function.
A result similar to Theorem~\ref{theorem:nlin} was sketched in \cite{Bloch1994:how}, with a logarithmic witness bound $v(n)$ rather than our stronger subpolynomial bound.
To extract the most out of the iterated guessing technique, we have found that it is crucial (as we have done) to carefully take into account how the witness size grows as simulations are composed.

\section{Effective Guessing Would Imply a SAT\\ Breakthrough}
\label{sec:egh:weak}

We now show that if general computations can be significantly sped up by using nondeterministic guessing to replace part of the computation, then this would imply a breakthrough for solving SAT.
More precisely, we show that using guessing to obtain an at least logarithmic factor reduction in time would imply that SAT can be decided in linear time on a nondeterministic multitape Turing machine.\footnote{The question of whether SAT can be solved in linear time on a nondeterministic multitape Turing machine has previously been discussed within the cstheory stackexchange \mbox{community \cite{Wehar2015:question}.}}
Simple nondeterministic Turing machine algorithms for SAT use $O(n(\lg n)^c)$ time, for some constant $c\ge 1$, but this bound has resisted improvement for several decades and the at least logarithmic factor has stubbornly remained~\cite{Schnorr1978:satisfiability,Santhanam2001:lower}.

A high level sketch of the argument for proving Theorem~\ref{theorem:sat:nlin} is as follows.
First, we show that any $n$-bit CNF formulas (in a reasonable encoding) can have at most $4n/\lg n$ variables.
Then we establish a fairly precise time bound for sorting on deterministic multi-tape Turing machines: a list of $m$ integers, each of size $\lg n$, can be sorted in at most $O(m(\lg m)(\lg n))$ steps.
Combining these results we can show that \sat is in $\tiwi(n\lg n,4n/\lg n)$.
Now if $\dtime(n\lg n)$ were contained in $\ntime(n)$, then $\tiwi(n\lg n,4n/\lg n)$ would be contained in $\ntime(n)$, and therefore $\sat \in \ntime(n)$.
We proceed by proving technical results which will allow us to formalize this argument.

The following lemma shows that a nontrivial speedup of deterministic computation would also allow computations with a significant nondeterministic component to be sped up.

\begin{lemma}[Weak Speedup]
\label{lemma:speedup:weak}
If $\dtime(t(n)) \subseteq \ntime(n)$, then
 \[
  \tiwi(t(n),n) \subseteq \ntime(n).
 \]
\end{lemma}

\begin{proof}
Suppose $\dtime(t(n)) \subseteq \ntime(n)$ and let $J \in \tiwi(t(n),n)$.
Then there is some $K \in \dtime(t(n))$ such that $J = K[n]$.
By our assumption, $K \in \ntime(n)$.
By Lemma~\ref{lemma:tiwi:ntime} it follows that there is some $c>0$ such that $K \in \tiwi(n,cn)$.
Hence there is some $L \in \dtime(n)$ such that $K = L[cn]$.
We conclude via Lemma~\ref{lemma:project:composed} that $J = K[n] = L[cn,n] = L[n+2cn]$, so again by Lemma~\ref{lemma:tiwi:ntime}, $J \in \tiwi(n,(2c+1)n) \subseteq \ntime(n)$.
\end{proof}

Although Lemma~\ref{lemma:speedup:weak} is closely related to Lemma~\ref{lemma:speedup:strong}, the weaker hypothesis of the Weak Speedup Lemma means that these results are not directly comparable.

\subsection{Improved Algorithms For SAT From Effective\\ Guessing}

We now apply the Weak Speedup Lemma to show that effective guessing implies improved algorithms for \sat.

\begin{corollary}
\label{cor:egh:bettersat}
Suppose $\sat \in \tiwi(t(n),n)$ for some function $t \colon \N \to \N$.
If further $\dtime(t(n)) \subseteq \ntime(n)$, then $\sat \in \ntime(n)$.
\end{corollary}

\begin{proof}
From Lemma~\ref{lemma:speedup:weak} it follows that if $\dtime(t(n)) \subseteq \ntime(n)$ then $\tiwi(t(n),n)$ $\subseteq \ntime(n)$.
\end{proof}

In Corollary~\ref{cor:egh:bettersat}, the time upper bound $t(n)$ for \sat enables the efficient guessing hypothesis to yield an improved algorithm for \sat.
Classical results imply that $\sat \in \tiwi(n(\lg n)^c,n)$ for some unspecified constant $c$.
This is because a guess-and-check procedure can be implemented via sorting~\cite{Schnorr1978:satisfiability}, and the number of variables determinines the witness size yet cannot exceed the input size.
We could therefore conclude a linear time upper bound for \sat from an effective guessing hypothesis of the form $\dtime(n(\lg n)^c) \subseteq \ntime(n)$.

A smaller time bound for \sat permits a weaker effective guessing hypothesis.
How weak can the hypothesis be made?
It turns out that we can actually take $c=1$ with some additional work.
This sharper bound requires two ingredients.

The first ingredient is a reasonable encoding of SAT, which distinguishes between formulas in conjunctive normal form (CNF) which only differ by a permutation of their variable names.
Reasonable encodings are used in Cook's original proof of the Cook--Levin theorem~\cite{Cook1971:complexity} and Karp's list of 21 NP-complete problems~\cite{Karp1972:reducibility}.
In fact, we are not aware of any work which relies on a particular encoding of \sat yet does not use a reasonable encoding of \sat.
Furthermore, the standard DIMACS CNF encoding used by SAT solvers\footnote{See \url{http://archive.dimacs.rutgers.edu/pub/challenge/satisfiability/doc/}.} also qualifies as reasonable.
We show that a reasonable encoding of \sat has the property that an $n$-bit CNF formula cannot represent more than $O(n/\lg n)$ different variables, eventually.

Second, we need sharp time bounds for sorting on a Turing machine.
Standard mergesort algorithms are slightly wasteful when implemented on a Turing machine, so we take a closer look at Schnorr's classical approach (from~\cite{Schnorr1978:satisfiability}) to obtain a more precise time bound.

The saving in the witness size due to a reasonable encoding is offset by overhead from sorting, but
combining these two ingredients allows us to conclude $c=1$.

\subsection{Bounding the Number of Variables in a CNF Formula}
\label{sec:bounding}

The following technical lemmas will be used to relate the number of variables in a SAT instance to its size.

\begin{lemma}
\label{lemma:logs}
Suppose $x_0>0$ and $v$ is a real-valued function which satisfies the inequality $v(x)\lg v(x) \le x$ for all $x\ge x_0$.
Then for every $C > 1$ there is some $x_1 > 1$ such that for every $x \ge x_1$, $v(x) < Cx/\lg x$.
\end{lemma}

\begin{proof}
Instead of working with $v(x)$ such that $v(x)\lg v(x) \le x$, let's work with an extremal function $w(x)$ such that $v(x) \le w(x)$ and $w(x)\lg w(x) = x$ for all $x \ge x_0$.
Further, put $w(x) = 2^{k(x)}$ and $x = w(x)\lg w(x) = k(x)2^{k(x)}$.
Then $w(x)/(x/\lg x) = \lg x/\lg w(x) = (k(x)+\lg k(x))/k(x) = 1 + (\lg k(x))/k(x)$, which tends to 1 as $k(x) \to \infty$ (which coincides with $x \to \infty$).
However, this expression is strictly greater than 1 for $k(x)>1$, i.e.~for $w(x) > 2$.
Eventually the fraction becomes arbitrarily close to 1, so we can say that eventually $w(x) < Cx/\lg x$ for any $C > 1$, and the result follows.
\end{proof}

\begin{lemma}
\label{lemma:logsmain}
Suppose $0<d<1$ and $x_0 > 0$.
Further, suppose that $v$ is a real-valued function such that $v(x) \ge 0$ and $(1-d)v(x)\lg v(x) \le x$ for all $x\ge x_0$.
Then for every $C \ge 4$ there is some $x_1 > 1$ such that for every $x \ge x_1$, $v(x) < Cx/\lg x$.
\end{lemma}

\begin{proof}
Given $C \ge 4$, let $d$ be the smaller of the two solutions of the equation $d(1-d) = 1/C$.
Since $0<1/C\le 1/4$, we have that $0<d \le 1/2$, and it follows that $\lg(1-d) < 0$.
Let $w(x) = (1-d)v(x)$.
Then for all $x\ge x_0$, $w(x)\lg w(x) = (1-d)v(x)(\lg(1-d) + \lg v(x)) < (1-d)v(x)\lg v(x) \le x$.
Since $1/d \ge 2 > 1$, by Lemma~\ref{lemma:logs} eventually $w(x) < (1/d)x/\lg x$.
Therefore eventually $v(x) < x/(d(1-d)\lg x) = Cx/\lg x$.
\end{proof}

We now assert that an encoding of \sat which removes all symmetries due to variable names does not constitute a reasonable encoding.
An unreasonable encoding could represent a CNF formula by a binary encoded integer which represents one particular CNF formula out of a predetermined list of equivalence classes of CNF formulas, with formulas regarded as equivalent up to reordering and renaming of variables.
We instead consider only reasonable encodings, which have the property that if two CNF formulas can be obtained from each other by simply permuting variable names, then these formulas will be represented by different words in the language.
With this restriction on what constitutes a reasonable encoding of SAT, we now prove an upper bound on how many variables can appear in a SAT instance in terms of its size.

\begin{lemma}
\label{lemma:satliterals}
In any reasonable encoding, $n$-bit CNF formulas eventually contain at most $4n/\lg n$ distinct variables.
\end{lemma}

\begin{proof}
Suppose $x$ is an $n$-bit input.
We are only interested in inputs that are valid CNF formulas, so further suppose that $x$ represents a propositional formula in CNF,
and that this formula uses $v$ distinct variables.
We will show that $v \le 4n/\lg n$ eventually.

Let $s_x$ be the $v$-element sequence formed by listing the first occurrence of each variable in the formula encoded by $x$.
(Note that $v$ depends on $x$.)
Any reasonable representation must be able to distinguish each of the $v!$ possible ways that $s_x$ can occur, one for each permutation of the variables.
Hence at least $\lg(v!)$ bits are required in the worst case, for any reasonable encoding of SAT.
By the Robbins bounds~\cite{Robbins1955:remark}
 \[
  \lg v! = v\lg v - v\lg e + (1/2)\lg(2\pi) + (1/2)\lg v + r_v
 \]
where $(1/(12v+1))\lg e < r_v < (1/12v)\lg e$,
and so for $n = \size{x}$ we have
 \[
  v\lg v - v\lg e + (1/2)\lg(2\pi) + (1/2)\lg v + (1/(12v+1))\lg e < n.
 \]
Hence for any $0 < d < 1$, eventually
 \(
  (1-d)v\lg v < n.
 \)
By Lemma~\ref{lemma:logsmain} there is then some $n_1 \ge n_0+1$ such that for all $n \ge n_1$, $v < 4n/\lg n$.
\end{proof}

\subsection{A More Precise Sorting Time Bound}
\label{sec:sorting:bound}

Results about sorting on a Turing machine are used in many classical papers.
However, as far as we are aware, a time bound has not been expressed in the literature in the precise form that we will present here.
We do not claim originality for such a result, but also have not been able to locate a proof with this bound.
We therefore provide a proof for completeness.

\begin{lemma}
\label{lemma:sorting}
A deterministic multitape Turing machine can sort a list of $m$ non-negative integers, each represented in binary encoding using \zap{at most }$\lg n$ bits, in $O(m(\lg m)(\lg n))$ steps.
\end{lemma}

\begin{proof}
We use a form of bottom-up mergesort.
Instead of a random access algorithm such as that of Batcher~\cite{Batcher1968:sorting}, we use a procedure that uses a fixed number of tapes and only sequential access, and can therefore be efficiently implemented on a deterministic multitape Turing machine.
This algorithm is a more detailed version of that outlined by Schnorr~\cite[Program $p_1$]{Schnorr1978:satisfiability}.
These additional details allow a more precise analysis of the time bound, which Schnorr was not attempting to optimise.

The algorithm proceeds in stages.
At each stage we use three tapes containing permutations of the list of $m$ integers: Result, Source, and Target.
During each stage, Source and Target are piecewise merged to form Result.
Result then becomes the Source for the next stage, and is copied to Target to begin the next stage.
Half the elements to be merged in each stage are on Source and the other half on Target: the actual contents of Source and Target are identical but we pay attention to a different pattern of sequences on Source compared to Target.
We use two copies of the list (one on Source and one on Target), rather than a single source tape, to avoid back-and-forth tape head moves.
This is key to keeping the runtime under control.

After $\lceil \lg m \rceil \le 2\lg m$ stages the current Result tape contains a sorted list.
Moreover, each stage uses $O(m\lg n)$ steps.
This is because we can use a small fixed number of tapes to keep track of various unary quantities, and two tapes as temporary workspace to copy the integers on Source and Target that are the current focus of attention.
This allows the machine to move the heads on Result, Source, and Target tapes only in one direction, with no backward motion required.
Backward motion is only used when the heads are repositioned to the start of each tape, at the end of each stage.
Moreover, the head movements on the auxiliary tapes only require a constant factor overhead.
The overall time bound then follows.

We first pad the input with dummy values that represent a number larger than the largest integer represented using $\lg n$ bits, so that the number of values in the list is a power of $2$ (and, in particular, $\lg m$ is a non-negative integer).
The overhead of this padding stage is included in the unspecified constant factor in the overall time.
(Moreover this also only increases the space used by at most a factor of 2.)
We now outline the key steps for the case where $m$ is a \mbox{power of $2$.}

For each $i=0,1,\dots,(\lg m) - 1$, if Source and Target at the start of stage $i$ contain $m/2^i$ sequences of sorted sublists, each sublist of length $2^i$, then at the end of the stage Result will contain $m/2^{i+1}$ sorted subsequences, each containing $2^{i+1}$ elements.
At stage $i$, the Source tape head is at the start of the tape, and we move the Target tape head to the position after the $2^i$th entry in the list (position $2^i\lg n$ if the first position on the tape is numbered $0$).
Once the first sublist has been processed, we move the heads forward by $2^i\lg n$ positions.
We proceed until the Target tape head reaches the position after the end of the whole list, position $m\lg n$.

To process a single pair of sublists S and T (on the Source and Target tapes, respectively), we first set up a unary counter using an auxiliary tape to keep track of the length of these sublists, then scan the elements sequentially and write the sublist formed by merging S and T to the Result tape.
At each step we are deciding which of a pair of elements to write to the result tape.
We write the smaller of the two current elements to the Result tape.
We do this by copying the current elements to two auxiliary tapes, and during copying flagging which tape contains the smaller of the two elements.
The auxiliary tapes are then rewound, and the flagged tape is copied to the Result tape.
\end{proof}

Schnorr proved a time bound of $O(m(\lg m)^c(\lg n))$ steps for some unspecified $c$~\cite{Schnorr1978:satisfiability}.
By a more detailed analysis of the tape head motion than was considered in Schnorr's argument we have obtained this more precise exponent for the logarithmic factor of $c=1$.

\subsection{Improving the Time-Witness Bound For SAT}

We are now able to prove a time-witness upper bound on \sat.

\begin{lemma}
\label{lemma:sattiwi}
$\sat \in \tiwi(n\lg n,4n/\lg n)$.
\end{lemma}

\begin{proof}
By Lemma~\ref{lemma:satliterals}, \sat can be decided nondeterministically by guessing an assignment to the eventually at most $4n/\lg n$ variables, and then verifying that the assignment satisfies the input formula.
The verification is deterministic, and can be done by first making a copy of the input formula while annotating every literal with a clause number, then sorting the literals by variable identifier, replacing each literal by its value in the guessed assignment, and finally sorting the values by clause number and scanning to check that at least one literal in each clause is set to true.
This procedure is a special case of the more general algorithm suggested by Schnorr~\cite[Program $p_3$]{Schnorr1978:satisfiability}.
The augmented formula is at most twice as long as the original, and by Lemma~\ref{lemma:sorting} it can be sorted in at most $O((n/\lg n)\,(\lg n)^2)$ steps, which is $O(n\lg n)$ steps.
It follows that $\sat \in \tiwi(n\lg n,4n/\lg n)$.
\end{proof}

We restate this in terms of the $\ntigu$ notation.

\satbound*

\begin{proof}
Follows immediately from Lemmas~\ref{lemma:sattiwi} and \ref{lemma:tiwi:ntigu}.
\end{proof}

Our time-witness bound for \sat then yields the main result of this section.

\weakegh*

\begin{proof}
By Lemmas~\ref{lemma:sattiwi} and \ref{lemma:tiwi:increase}, we conclude that $\sat \in \tiwi(n\lg n,4n/\lg n) \subseteq \tiwi(n\lg n,n)$.
Now we can apply Corollary~\ref{cor:egh:bettersat}, and so if $\dtime(n\lg n) \subseteq \ntime(n)$ then $\tiwi(n\lg n,n) \subseteq \ntime(n)$.
\end{proof}

\section{Conclusion and Further Work}
\label{sec:misc}

Our contributions in this work demonstrate that effective guessing has unlikely consequences.
We therefore propose an \emph{ineffective guessing conjecture}, that it is not in general possible to speed up a computation significantly by using more nondeterminism.

More precisely, we propose the following ineffective guessing conjecture:

\begin{conjecture}[IGC]
\label{conjecture:igc}
$\dtime(t(n)) \not\subseteq \ntime(n)$ for all time-constructible functions $t(n)$ such that $t(n) = \omega(n\lg n)$.
\end{conjecture}

According to this ineffective guessing conjecture, it is not in general possible to obtain even a slightly greater than logarithmic speedup by making essentially every step of a computation nondeterministic.  Furthermore, our ineffective guessing conjecture implies that the effective guessing hypothesis from Theorem~\ref{theorem:p:nlin}, with a polynomial speedup, is too strong while the weaker effective guessing hypothesis from Theorem~\ref{theorem:sat:nlin} could still hold.

To put Theorems \ref{theorem:nlin} and \ref{theorem:sat:nlin} into context,
the effective guessing hypotheses used in these theorems fall between two extremes.
\begin{align*}
\textbf{excessively-weak EGH}   &: \;
  \dtime(n) \subseteq \ntigu(n, 0)
\\
\textbf{weak EGH}               &: \;
  \dtime(n\lg n) \subseteq \ntime(n)
\\
\textbf{strong EGH}             &: \;
  (\exists c > 1) \; (\forall d > 0) \; \dtime(n^c) \subseteq \ntigu(n, n^d)
\\
\textbf{excessively-strong EGH} &: \;
  (\exists c > 0) \; \dtime(n^{2+c}) \subseteq \ntigu(n,\lg n)
\end{align*}
The hypothesis from Theorem~\ref{theorem:nlin} is the strong EGH, while Theorem~\ref{theorem:sat:nlin} posits the weak EGH.
To be clear, both of these hypotheses currently remain open, although we have shown that they have somewhat unlikely consequences.

Excessively weak forms of effective guessing are always true such as $\dtime(n)$ $\subseteq \ntigu(n, w(n))$ which holds for any function $w(n)$, even ${w(n)=0}$.
This therefore forms one extreme, a hypothesis about effective guessing that is too weak to be interesting.
On the other hand, we show in the following that excessively strong forms of effective guessing (such as that stated above) can be ruled out unconditionally.

\begin{lemma}
\label{lemma:brute:consequence}
If $t(n) = \omega(n^2)$ is a function that is computable in $t(n)$ steps, then \[
  \dtime(t(n)) \not\subseteq \ntigu(n,\lg n).
\]
\end{lemma}

\begin{proof}
Suppose $t(n) = \omega(n^2)$ such that $t(n)$ is computable in $t(n)$ steps and $\dtime(t(n)) \subseteq \ntigu(n,\lg n)$.
By trying all $2^{\lg n} = n$ possible values for the witness and checking each in $O(n)$ time we have $\ntigu(n,\lg n) \subseteq \dtime(n^2)$.
Thus $\dtime(t(n)) \subseteq \dtime(n^2)$.
As $t(n) = \omega(n^2)$ this then contradicts the deterministic time hierarchy theorem~\cite{Furer1982:tight}.
\end{proof}

\begin{proposition}
\label{proposition:brute:consequence}
$\dtime(n^{2+c}) \not\subseteq \ntigu(n,\lg n)$ for all $c>0$.
\end{proposition}

\begin{proof}
The result follows from Lemma~\ref{lemma:brute:consequence} for $t(n) = n^{2+c}$.
\end{proof}

Since the strong EGH trivially implies the weak EGH, we can therefore rank the hypotheses in terms of logical strength as follows:

\medskip

\begin{centering}

\begin{tabular}{cl}
\textbf{excessively-strong EGH} & \textbf{ [false]}
\\
$\Downarrow$
\\
\textbf{strong EGH}             & \textbf{ [open]}
\\
$\Downarrow$
\\
\textbf{weak EGH}               & \textbf{ [open]}
\\
$\Downarrow$
\\
\textbf{excessively-weak EGH}   & \textbf{ [true]}
\end{tabular}

\end{centering}

\medskip

Finally, our effective guessing hypotheses focus on nondeterministic linear time because the Paul et al.~\cite{Paul1983:determinism} result that $\dtime(n) \subsetneq \ntime(n)$ invites many questions about the potential computational power of nondeterministic linear time.
A natural future direction would be to consider the computational power of $\ntime(n^k)$ for $k > 1$.
In particular, it is still is not known whether $\dtime(n^k) \subsetneq \ntime(n^k)$ for any $k > 1$.
Furthermore, there are many additional open questions such as whether any languages exist in $\ntime(n^k) \setminus \ntigu(n^k,o(n^k))$.

\subsection*{Acknowledgments}

The authors thank Ryan Williams, Rahul Santhanam, and Kenneth Regan for useful discussions.  We also acknowledge the helpful discussion and comments from \cite{Wehar2015:question} which have helped us to provide a detailed treatment of Theorem~\ref{theorem:sat}.

\subsection*{Funding}

The first author's work was supported by EPSRC grant EP/P015638/1.

\bibliographystyle{plainurl-modern}
\bibliography{references}

\end{document}